\documentclass[a4paper,11pt,english]{article}
\usepackage[utf8]{inputenc}

\usepackage{enumerate}
\usepackage{multirow}
\usepackage{calc}
\usepackage{graphicx,palatino,verbatim}
\usepackage{amsmath}
\usepackage[T1]{fontenc}
\usepackage[english]{babel}
\usepackage{newlfont}
\usepackage{amsthm}
\usepackage{amssymb}
\usepackage{amsmath}
\usepackage{latexsym}
\usepackage[mathscr]{eucal}
\usepackage{amsfonts}
\usepackage{syntonly}
\usepackage{graphicx}
\usepackage{mathptmx}
\usepackage{subfig}
\usepackage{url}
\usepackage[a4paper]{geometry}
\usepackage{xcolor}
\usepackage{algorithm}
\usepackage{algpseudocode}
\algnewcommand\algorithmicinput{\textbf{Input:}}
\algnewcommand\INPUT{\item[\algorithmicinput]}
\algnewcommand\algorithmicoutput{\textbf{Output:}}
\algnewcommand\OUTPUT{\item[\algorithmicoutput]}
\newtheorem{theo}{Theorem}

\newtheorem{remark}[theo]{Remark}

\newtheorem{definition}[theo]{Definition}

\newtheorem{proposition}[theo]{Proposition}

\newtheorem{assumption}{Statistical Assumption}
\DeclareMathOperator{\ord}{ord}
\DeclareMathOperator{\LCM}{LCM}
\DeclareMathOperator{\GCD}{GCD}
\DeclareMathOperator{\CRT}{CRT}

\DeclareMathOperator{\DLog}{DLog}
\newcounter{mycount}

\title{A Discrete Logarithm-based Approach to Compute Low-Weight Multiples of Binary Polynomials}
\author{P. Peterlongo, M. Sala and C. Tinnirello\thanks{\textit{Corresponding author}: claudia.tinnirello@unitn.it}\\
\\
\textit{Department of Mathematics, University of Trento, Italy}}

\date{\today}

\begin{document}

\maketitle

\begin{abstract}
Being able to compute efficiently a low-weight multiple of a given binary polynomial is often
a key ingredient of correlation attacks to LFSR-based stream ciphers.
The best known general purpose algorithm is based on the generalized birthday problem.
We describe an alternative approach which is based on discrete logarithms and has much lower memory complexity requirements
with a comparable time complexity.
\end{abstract}

\section{Introduction}
A Linear Feedback Shift Register (LFSR) is the basic component of many keystream generators for stream ciphers applications. It is
defined by its connection polynomial, which is a binary polynomial.
A parity check for a single LFSR is a multiple of its connection polynomial, while
a parity check for more than one LFSR is a multiple of the least common multiple of the respective connection polynomials.
The weight of a parity check is the weight of the associated polynomial, that is, the number of its nonzero coefficients.

Correlation attacks were introduced by Siegenthaler in \cite{siegenthaler1985decrypting} to cryptanalyze
a large class of stream ciphers based on LFSRs.
A major improvement by Meier and Staffelbach \cite{meier1989fast} led to different versions of
fast correlation attacks \cite{canteaut2000improved,johansson1999improved,chepyzhov2001simple}.

These attacks try to find a correlation between the output of the stream cipher and
one of the LFSRs on which it is built, then they try to recover the state of the LFSR by decoding the keystream
as a noisy version of the LFSR output.
A fast version of a correlation attack involves the precomputation of multiple
parity checks of one of the LFSRs in order to speed up the computation.
This precomputation step can be computed according to two
different approaches, one based on the birthday paradox \cite{chose2002fast} and another based on
discrete logarithms \cite{didier2007finding}.
The determination of an efficient algorithm for the computation of a polynomial multiple
(subject to given constraints, e.g. on the degree) is an interesting problem in itself for computational algebra.

It might be argued that the design of modern stream ciphers evolved accordingly. A cipher like E0 (\cite{gehrmann2004bluetooth})
is not subject to these types of attacks, since no single LFSR is correlated to the keystream output.
In \cite{lu2004faster}, Lu and Vaudenay introduced a new fast correlation attack which
is able to successfully recover the state of E0.
Their attack requires a precomputation step which computes a \emph{single} parity check
of \emph{multiple} LFSRs.
The complexity of the precomputation step is not far from the complexity of their full attack,
and in order to keep this estimate low, they employed
the generalized birthday approach presented in \cite{wagner2002generalized}.
Further research was recently presented in \cite{pts2014yacc}, which contains in particular a straightforward generalization of the
discrete logarithm approach of \cite{didier2007finding}.

In this paper we generalize the discrete log approach of \cite{didier2007finding},
proposing an algorithm that is able to compute a single parity check of multiple LFSRs. To be more precise, the problem we will address throughout
the paper is the following:
\begin{myprob}
 {Find a given-weight polynomial multiple with target degree}
 {A polynomial $p$ and two integers $w\ge 3$ and $D$.}
 {A multiple of $p$ of weight $w$ and degree at most $D$, if it exists.}
\end{myprob}

In Section \ref{sec:math}, we present our strategy, we fix the notation and we give algebraic results on which the algorithm is based.
The algorithm we propose is explained in Section \ref{sec:algorithm}, along with a comparison of its complexity
to the generalized birthday approach and to the straightforward generalization of the discrete log approach for the case
of a single primitive polynomial.
A significant example application of our approach is outlined in Subsection \ref{sss:e0}, where we show that for the fast
correlation attack in \cite{lu2004faster} our algorithm could be more convenient to use in the second precomputation step than
the generalized birthday approach.
We draw our conclusions in Section \ref{sec:conclusions}.

\section{Our Strategy} \label{sec:math}

In this section we explain our general strategy for the resolution of Problem 1, making reference
to the results reported in Subsection \ref{sbs:results}.
We can translate the problem of finding a multiple of a binary polynomial
into finding an appropriate sequence of integers modulo the order of a polynomial (Propositions~\ref{pro:order},\ref{pro:multiple_criterion}).
Next, we show that the general problem is solved once we know how to solve the problem for single factors
of the unique factorization into irreducible polynomials (Proposition~\ref{pro:coprimes}).
The case of irreducible non-primitive polynomials is reduced to the case of primitive polynomials, which is
solved through discrete logarithms by an explicit construction which uses Zech's logaritms (Propositions~\ref{pro:vanish_in_alpha},\ref{pro:primitive},\ref{pro:irreducible_not_primitive}).
Finally, powers of irreducible polynomials could be completely characterized, although in our algorithm we will
aim for a subset of all possible multiples (Remark~\ref{rem:powers_of_irreducibles}) for computational reasons.

%

\subsection{Notation}

In an Euclidean ring, we denote by $a \bmod c$ the remainder of the division of $a$ by $c$.
We write $a \equiv b \pmod c$ if $a$ is congruent to $b$ modulo $c$.
We denote by $\LCM$ and $\GCD$
the Least Common Multiple and the Greatest Common Divisor  --- respectively ---
of polynomials or integers, depedending on its inputs.

Let $R = \mathbb{F}_2[x]$ be the ring of binary polynomials.
A binary polynomial is uniquely determined by the position of its nonzero coefficients, that is, its support.
We denote a binary polynomial by its exponents:
\[
[e_1, e_2, \ldots, e_k] := x^{e_1}+x^{e_2}+\cdots + x^{e_k},
\]
where the $e_i$'s are positive integers. If two integers in the list
are repeated we can omit them from the list and obtain the same binary polynomial.
If no integers are repeated then $[e_1,\ldots,e_k]$ is a polynomial of weight $k$.
From now on when we say ``polynomial'' we actually mean ``binary polynomial''.

The remainder of the division of
a polynomial by the polynomial $1+x^N$
is obtained through reduction modulo $N$ of the exponents:
\begin{equation} \label{eq:remainder}
[e_1,e_2, \ldots, e_k] \bmod [0,N] = [\bar{e}_1,\bar{e}_2,\ldots,\bar{e}_k],
\end{equation}
where $0 \leq \bar{e}_i < N$ and $e_i \equiv \bar{e}_i \pmod N$.

We will denote by $\CRT(e^{(1)},\ldots, e^{(r)},N_1,\ldots, N_r)$ the result of applying the Chinese Remainder Theorem
to integers $e^{(i)}$ and moduli $N_i$. 
This is defined if and only if
\[
e^{(i)} \equiv e^{(j)} \pmod {\GCD(N_i,N_j)}
\]
for all $i,j$ and the result will be the unique $0 \leq \bar{e} < \LCM(N_1,\ldots,N_r)$
such that $\bar{e} \equiv e^{(i)} \pmod {N_i}$ for each $i$.

\subsection{Preliminary Results} \label{sbs:results}

\begin{definition}[{order of a polynomial, \cite[Definition~3.2]{CGC-cd-book-niederreiter97}}]
Let $p$ be a non-zero polynomial in $\mathbb{F}_2[x]$. 
If $p(0)\neq0$, then the order
of $p$ is the least positive integer $N$ such that $1+ x^N$ is a multiple of $p$, and
we denote it by $\ord(p)$.
\end{definition}

\begin{proposition}[{\cite[Theorem~3.11, Theorem~3.16, Corollary~3.4]{CGC-cd-book-niederreiter97}}] \label{pro:order}
Let $p$ be a binary polynomial of positive degree and $p(0)\neq 0$. 
Let $p=p_1^{b_1}\cdots p_r^{b_r}$, where $b_1,\ldots, b_r\in\mathbb{N}$ and $p_1,\ldots,p_r$
are distinct irreducible polynomials of degree $n_1,\ldots,n_r$,
be the factorization of $p$ in $\mathbb{F}_2[x]$.
Then,
\begin{enumerate}
 \item $\ord(p)=2^t \LCM(\ord(p_1),\ldots, \ord(p_r))$ where $t$ is the smallest integer
  such that $2^t$ is bigger or equal than $\max(b_1,\ldots, b_r)$;
 \item If $p_i$ is irreducible, then $\ord(p_i)\mid 2^{n_i}-1$;
 \item $p_i$ is primitive if and only if $\ord(p_i)= 2^{n_i}-1$.
\end{enumerate}
\end{proposition}

\begin{proposition}
\label{pro:multiple_criterion}
Let $f,f',p$ be polynomials, $f=[e_1,\ldots,e_k]$ and $f'=[e_1',\ldots,e_k']$ such that for any~$i$,
$e_i \equiv e_i' \pmod M$, where $M$ is a multiple of $\ord(p)$.
Then, $f$ is a multiple of $p$ if and only if $f'$ is so.
\end{proposition}

\begin{proof}
Since $e_i \equiv e_i' \pmod M$, by \eqref{eq:remainder}, the remainder of the division of $f$ and $f'$ by $1+x^M$ is the same. Let $r$ be this remainder. 
So $f=q (1+x^M)+r$ and $f'=q' (1+x^M)+r$ for some $q$ and $q'$. Let us suppose that $f$ is a multiple of $p$. Since $M$ is a multiple of $\ord(p)$,
$1 + x^M$ is a multiple of $p$.
Then $r$ is also a multiple of $p$, and so $f'$ is a multiple of $p$.  
\end{proof}

\begin{proposition} \label{pro:coprimes}
 Let $p=g_1 \cdots g_r$ with $\GCD(g_i,g_j)=1$ for all $i\neq j$. We denote $N_i = ord(g_i)$ and $N=\LCM(N_1, \ldots, N_r)$.
 Given $w-2$ distinct integers $0<e_2< \ldots < e_{w-1}<N$,
 if there exist $r$ integers $e_{w}^{(1)}, \ldots e_{w}^{(r)}$ with $0 < e_{w}^{(i)} < N_i$ 
 such that for all $i$ and $j$ 
  \begin{enumerate}
  
    \item $[0,e_2,\ldots, e_{w-1}, e_{w}^{(i)}]$ is a multiple of $g_i$
     
    \item $e_{w}^{(i)}\equiv e_{w}^{(j)} \pmod {\GCD(N_i,N_j)}$
     
  \end{enumerate}

 then $[0,e_2,\ldots, e_{w-1}, e_{w}]$ is a multiple of $p$ of weight $w$, where
 \[
 e_{w} = \CRT(e_{w}^{(1)}, \ldots e_{w}^{(r)},N_1,\ldots, N_r).
 \]
 
Furthermore, all multiples of $p$ of weight $w$ and degree at most $N$ can be obtained in this way, that is,
if $[0,e_2,\ldots, e_{w-1}, e_{w}]$ with $0 < e_2<\ldots< e_{w} < N$ is multiple of $p$ of weight $w$,
then the integers $e_{w}^{(i)} = e_w \bmod {N_i}$ satisfy points 1 and 2 above.
\end{proposition}

\begin{proof} 
Let $j\in\{1,\ldots, r\}$ and let us prove that $[0,e_2,\ldots, e_{w-1}, e_{w}]$ is a multiple
of $g_j$. By point 2 we can apply CRT and get an integer $e_{w}< N$
such that $e_{w}\equiv e_{w}^{(i)} \pmod {N_i}$ for all $i$. Then, by point 1, applying Proposition \ref{pro:multiple_criterion} to $f=[0,e_2,\ldots, e_{w-1}, e_{w}]$ and $f'=[0,e_2,\ldots, e_{w-1}, e_{w}^{(j)}]$, we get that $f$ is a multiple of $g_j$.\\
Let $[0,e_2,\ldots, e_{w-1}, e_{w}]$ be a multiple of $p$ of weight $w$ of degree at most $N$, and for $i\in\{1,\ldots, r\}$ let $e_{w}^{(i)} = e_w \bmod {N_i}$. Then, for any $i$, 
$[0,e_2,\ldots, e_{w-1}, e_{w}]$ is a multiple of $g_i$, and so, by Proposition \ref{pro:multiple_criterion} the integers
$e_{w}^{(i)}$'s satisfy point 1. On the other hand, thanks to CRT, since $e_{w}^{(i)} = e_w \bmod {N_i}$, they satisfy also point 2.
\end{proof}

\begin{proposition}[{\cite[Lemma~2.12 (ii)]{CGC-cd-book-niederreiter97}}] 
\label{pro:vanish_in_alpha}
Let $p$ be an irreducible polynomial and let $\alpha$ be a root of $p$ in an extension field
of $\mathbb{F}_2$. Then a polynomial $q$ is a multiple of $p$ if and only if $q(\alpha)=0$.
\end{proposition}

\begin{definition}[Discrete Logarithm] 
Let $\alpha \in \mathbb{F}_{2^n}$ a primitive element and $\beta  \in \mathbb{F}_{2^n}, \beta \neq 0$. We define
the discrete logarithm of $\beta$ with respect to $\alpha$ as the unique integer $0 \leq i < 2^n$ such that
$\alpha^i=\beta$. We use the notation $i = \DLog_\alpha(\beta)$.
\end{definition}

\begin{remark}
Let $\alpha$ be the root of a primitive polynomial $p$.
Given $e_1,\ldots, e_{w-1}$, if we are able to compute
\begin{equation} \label{eq:dlog}
e_w = \DLog_\alpha(\alpha^{e_1} + \ldots + \alpha^{e_{w-1}}),
\end{equation}
by Proposition \ref{pro:vanish_in_alpha}, we know that $[e_1,\ldots,e_{w-1},e_w]$ is a multiple of $p$.
Viceversa, if $[e_1,\ldots,e_{w-1},e_w]$ is a multiple of $p$, then \eqref{eq:dlog} holds.
\end{remark}

In order to give a characterization of the $(w-1)$-uples of exponents for which the expression
\eqref{eq:dlog} is computable, we introduce Zech's Logarithms.
They can also be used to compute in an efficient way discrete logarithms (see \cite{huber1990some,douillet2001zech}).

 \begin{definition}[Zech's Logarithm]
 Let $p$ be a primitive polynomial of degree $n$,  $\alpha$ a root of $p$ in $\mathbb{F}_{2^n}$, and
 $N=\ord(p)$.
 The Zech's Logarithm with base $\alpha$ of an integer $j$ is the integer $0 < i < N$ such that 
 $i = \DLog_\alpha(1 + \alpha^j)$ and it will be denoted by $i =Z_\alpha(j)$.
 When $j \bmod N = 0$, then $1+\alpha^i = 0$ and we will say that $Z_\alpha(j)$ is not defined.
 \end{definition}

\begin{remark} \label{rem:zlogs}
The expression \eqref{eq:dlog} can be rewritten in terms of a chain of Zech's Logarithms in this way:
\[
\alpha^{e_w}= \alpha^{e_1} + \ldots + \alpha^{e_{w-1}} = \alpha^{e_1} ( 1 + \alpha^{e_2 - e_1} ( \cdots ( 1+ \alpha^{e_{w-2} - e_{w-3}} (1 + \alpha^{e_{w-1} - e_{w-2}})) \cdots ))
\]
\begin{equation}
e_w = e_1 + Z_\alpha( e_2 - e_1 + Z_\alpha( \cdots Z_\alpha(e_{w-2} - e_{w-3} + Z_\alpha (e_{w-1} - e_{w-2}) ) \cdots ) )
\end{equation}
which can also be seen recursively as

\[
\begin{cases}
   E_1 = Z_\alpha (e_{w-1} - e_{w-2})\\
   E_{k+1} = Z_\alpha( e_{w-k-1} - e_{w-k-2} + E_k)
   
   \end{cases}
\]
and we denote the final result $e_w = e_1 + E_{w-2}$ by
\[
Z_\alpha(e_1,e_2,\ldots, e_{w-1}) = \DLog_\alpha(\alpha^{e_1} + \ldots + \alpha^{e_{w-1}}).
\]

Note that $Z_\alpha(e_1,e_2,\ldots, e_{w-1})$ is the concatenation of $(w-2)$ Zech's Logarithms.
Let us suppose that $E_i$ is not defined and $E_1, \ldots, E_{i-1}$ are defined. If $i = w-2$ (the last Zech's Logarithm is undefined)
then the discrete logarithm in \eqref{eq:dlog} cannot be computed since there is no integer $e_w$ such that $\alpha^{e_w}=0$. If $i<w-2$ we have that
\[
e_{w-i} - e_{w-i-1} + E_{i-1} \equiv 0 \pmod N
\]
\[
e_{w-i-1} \equiv e_{w-i} + E_{i-1} \pmod N
\]
\[
\alpha^{e_{w-i-1}} = \sum_{h=1}^{i}\alpha^{e_{w-h}}
\]
\[
\sum_{h=1}^{w-1}\alpha^{e_h} = \sum_{h=1}^{w-i-2}\alpha^{e_h}
\]
If $i=w-3$ the discrete logarithm is computed as $e_{w} = e_{1} \bmod N$.
Otherwise, we check whether we can compute the discrete logarithm $Z_\alpha(e_1,\ldots, e_{w-i-2})$ applying the same arguments.
\end{remark} 

An easy counting argument shows the following:
\begin{proposition} \label{pro:primitive}
Let $\alpha$ be a root of a primitive polynomial $p$ of order $N$.
The set of $(w-2)$-uples of integers $e_2,\ldots,e_{w-1}$
such that $0 < e_2 < \ldots < e_{w-1}<N$ has cardinality $\binom{N-1}{w-2}$.
Its subset such that we can compute $\DLog_\alpha(1 + \alpha^{e_2} + \ldots + \alpha^{e_{w-1}})$
has cardinality at least $\binom{N - 1}{w-3} (N - 2)$.
\end{proposition}

Note that for small $w$ and large $N$, $\binom{N - 1}{w-2}\simeq (N-1)^{w-2} \simeq N^{w-2}$  and $\binom{N - 1}{w-3}(N-2)\simeq (N-1)^{w-3}(N-2)\simeq N^{w-2}$. Then, for almost all $(e_2,\ldots,e_{w-1})$ we can compute $\DLog_\alpha(1 + \alpha^{e_2} + \ldots + \alpha^{e_{w-1}})$.

\begin{proposition} \label{pro:irreducible_not_primitive}
Let $p$ be an irreducible non-primitive polynomial of degree $n$ and order $N$, let $m=(2^n-1)/N$ and let
$\alpha$ be a primitive root in $\mathbb{F}_{2^n}$.
Then, $f$ is a multiple of $p$ with $f(0) \neq 0$ if and only if there exist integers $e_1, e_2,\ldots, e_{w-2}, e_{w-1}$
multiples of $m$ such that $0<e_1<e_2<\cdots< e_{w-2}<N$, $e_{w-1}=E_\alpha(e_1,e_2,\ldots,e_{w-2})$
with all Zech's logarithms defined, and $f=[0,d_1, d_2,\ldots, d_{w-2}, d_{w-1}]$ with $d_j=e_j/m$.    
\end{proposition}

\begin{proof}
Let $f=[0,d_1, d_2,\ldots, d_{w-2}, d_{w-1}]$ with $d_1, d_{2}, \ldots, d_{w-1}$ as in the statement
and let us prove that $f(\alpha^m)=0$. Since $\alpha^m$ is a root of $p$, by Proposition
\ref{pro:vanish_in_alpha} this is enough to guarantee that $f$ is a multiple of $p$.
On the other hand, $f(\alpha^m)=1+(\alpha^m)^{d_1}+(\alpha^m)^{d_2}+\ldots + (\alpha^m)^{d_{w-2}}+
(\alpha^m)^{d_{w-1}}= 1+\alpha^{e_1}+\alpha^{e_2}+\ldots \alpha^{e_{w-2}}+\alpha^{e_{w-1}}=0$. The last
equality follows from the hypothesis $e_{w-1}=E_\alpha(e_1,e_2,\ldots,e_{w-2})$.
Vice versa, let $f$ of weight $w$. Suppose that $f=[0,d_1, d_2,\ldots, d_{w-2}, d_{w-1}]$
with $0<d_1<d_2<\cdots< d_{w-2}<N$. If $f$ is a multiple of $p$, then, by Proposition \ref{pro:vanish_in_alpha},
$f(\alpha^m)=0$. From $f(\alpha^m)=0$, we get that $d_j=e_j/m$ for all $j$ and $e_{w-1}=E_\alpha(e_1,e_2,\ldots,e_{w-2})$.
\end{proof}

\begin{remark} \label{rem:powers_of_irreducibles}
 If $[e_1,\ldots,e_w]$ is multiple of $p$, then $[2^t e_1, \ldots, 2^t e_w]$
 is multiple of $p^b$ for all $b \leq 2^t$.
\end{remark}

In Algorithm 1 we will use Remark~\ref{rem:powers_of_irreducibles} to find multiples of powers of irreducibles.
Not all multiples can be found in this way unless the exponent for the repeated factor is a power of two.
In this case a complete characterization of multiples $f$ of power of irreducibles can be given using
the polynomial $\GCD(f,Df)$. This is not convenient in our case, since $\GCD(f,Df)$ has generally
a much higher weight than the weight of $f$.


\section{The Algorithm} \label{sec:algorithm}

In this section we propose an algorithm to solve Problem 1 
making reference to its pseudocode,
then we estimate its complexity and compare it to other approaches
. Finally, we consider the case of a correlation attack to E0 (Subsection \ref{sss:e0}).

\subsection{Description of Algorithm 1} \label{sbs:desc}

We describe our proposed algorithm making reference to the pseudocode in Algorithm \ref{algo:main}.

We take as input the factorization of the polynomial $p$:
\[
 p = p_1^{b_1} \cdot \ldots \cdot p_r^{b_r},
\]
where $p_1,\ldots,p_r$ are irreducible polynomials and $b_1,\ldots,b_r$ positive integers.
Computing the factorization is not computationally expensive, it can be efficiently computed using a probabilistic algorithm such as the Cantor-Zassenhaus algorithm
(cfr \cite{CGC-cd-book-niederreiter97}) and in some cases, such as in correlation attacks to LFSR-based stream ciphers,
the polynomial we are interested in is given already by its irreducible factors.

We start (line 2) by computing a root $\alpha_i^{m_i}$ for each irreducible polynomial $p_i$, expressed as a power of a primitive root $\alpha_i$ of $\mathbb{F}_{2^{\deg(p_i)}}$
(if the polynomial is primitive $m_i = 1$). Next (line 3) we compute the order of each root $\alpha_i^{m_i}$, we set $m$ as the
least common multiple of all $m_i$ (line 4), and we
set $t$ as an integer (line 5) that will allow us to take into account powers of irreducible polynomials (in line 12 we will output the multiple of $p$
as a multiple of $p_1 \cdot \ldots \cdot p_r$ elevated to $2^t / m$).

\begin{algorithm}
\begin{algorithmic}[1] 
  \Function{}{$p_1,\ldots,p_r,b_1,\ldots,b_r,w,D$}

    \Statex \Comment{All lines with $i$ repeat for $i=1,\ldots,r$}
    \State $\alpha_i, m_i \gets \mathrm{PrimitiveRoot}(p_i)$
    \Statex \Comment{If $p_i$ is primitive then $m_i =1$}
    \Statex \Comment{If $p_i$ is irreducible not primitive then $\alpha_i^{m_i}$ is a root of $p_i$}
    \State $N_i \gets (2^{\mathrm{deg}(p_i)} - 1)/m_i$
    \Statex \Comment{$N_i$ is the order of $\alpha_i^{m_i}$}
    \State $m \gets \LCM(m_1,\ldots, m_r)$
    \State $t \gets \mathrm{MinIntegerGreaterOrEqualThan}(\log_2 b_1, \ldots, \log_2 b_r)$

  \Repeat
    \State $e_{2}, \ldots, e_{w-1} \gets \mathrm{RandomDistinctMultiplesLessThan}(m \cdot D/2^t)$
    \Statex           \Comment{Random sampling of $(w-2)$ distinct integers $\leq D$}
    \Statex           \Comment{which are multiples of $m$}
    \State $e_{i,2}, \ldots, e_{i,k(i)} \gets \mathrm{ReduceAndShorten}(e_{2}, \ldots, e_{w-1}, N_i)$
    \Statex           \Comment{Reduce modulo $N_i$ and eliminate pairs of equal integers}
    \Statex           \Comment{we might obtain a shorter sequence of integers ($k(i) \leq w-1)$}
    \State $e_{i,w} \gets \mathrm{Z}_{\alpha_i} (0,e_{i,2}, \ldots, e_{i,k(i)})$
    \Statex           \Comment{If not possible, restart the cycle}
    \Statex           \Comment{also, restart the cycle if $e_{i,w} \bmod m_i \neq 0$}    
    \State $e_{w} \gets \mathrm{CRT}(e_{1,w},\ldots,e_{r,w}, N_1,\ldots, N_r)$
    \Statex           \Comment{If not possible, restart the cycle}

  \Until {$2^t e_{w} / m \leq D$}
  \State \Return $[0,2^t e_2 / m, \ldots, 2^t  e_w / m]$
  \EndFunction
 \end{algorithmic}
 \caption{Given the factorization of a polynomial $p$, $w$ and $D$, the algorithm finds (if it exists) a multiple of $p$ of weight $w$ and degree at most $D$.} \label{algo:main}
\end{algorithm}

The main cycle samples $w-2$ distinct integers less than $m \cdot D/2^t$ (line 7), then computes for each $i$ their remainders modulo $N_i$ and
discards pairs of equal elements (line 8).
The most demanding computation is done on line 9 where we compute the chain of Zech's logarithms
as described in Remark~\ref{rem:zlogs}. If we are not able to perform this computation or one of the $e_{i,w}$'s is not
a multiple of $m_i$, then we restart the cycle.
Otherwise, we try to compute the exponent $e_w$ through the Chinese Remainder Theorem (line 10) given the exponents computed at previous step.
If all $N_i$'s are coprime
we are sure to be able to perform this step, otherwise it might happen that a couple $(e_{i,w},e_{j,w})$ is not congruent modulo $\GCD(N_i,N_j)$.
In this case we restart the cycle. We remark that in the case of $N_i$ not all coprimes it may happen that we are \emph{never} able to compute this CRT.
We will comment on this aspect in \ref{sss:noncoprimes}.
The exit condition for the cycle (line 11) is verified if we have found a multiple of $p_1 \cdot \ldots \cdot p_r$ with degree at most $m \cdot D/2^t$,
from which we are easily able (line 12) to produce a multiple of $p$ of degree at most $D$ (which will have weight $w$ by construction).
The correctness of Algorithm 1 follows from the results cited in Section \ref{sec:math}.

\subsubsection{The case of $N_i$s not all coprimes} \label{sss:noncoprimes}

Suppose $p$ is the product of $p_1, p_2$ with $N_1=\ord(p_1), N_2=\ord(p_2)$ and that
$\GCD(N_1,N_2)$ is not $1$. In this case it may happen that we are never able to compute $\CRT(e_{1,w},e_{2,w},N_1,N_2)$, that
is, all possible pairs $(e_{1,w},e_{2,w})$ are not congruent modulo $\GCD(N_1,N_2)$.
Using bounds on the minimal distance of cyclic codes coming from coding theory (see \cite[Vol 1, p.60 and following]{CGC-cd-book-pless98}),
one can verify whether or not the $\CRT$ can be computed. If the minimal distance of the cyclic code generated by $p$
is $\geq w_0$, then obviously the $\CRT$ cannot be computed for all $w < w_0$. 

\subsection{Complexity estimates} 
\label{sbs:compare}
The complexity of Algorithm 1 is estimated in terms of the order of polynomial $N$ and target degree $D$ under an appropriate Statistical Assumption
(Subsection \ref{sss:complex1}). Then (Subsection \ref{sss:complex2}), we compare it to the complexity of birthday-based approaches expressing our estimates in terms of the complexity parameter $n$,
the degree of $p$.

\subsubsection{Complexity of Algorithm 1} \label{sss:complex1}

We set $n$, the degree of polynomial $p$, to be our complexity parameter.
Clearly, the main computational cost in the cycle is the $r$ computations of the chain of $(w-2)$ Zech's logarithms (line 8).
We denote the cost of the Zech's logarithm as a cost $C(n)$. Since we compute $r$ chains of $w-2$ Zech's logarithms,
in each cycle we have a complexity
of $O(r(w-2)C(n))$.
Observe that $C(n)$ depends on $n$ via dependence on the degrees $\delta_i=\deg(p_i)$. Among the
$\delta_i$'s there is one, which we may call $\Delta$, such that $2^{\Delta}-1$ has the biggest factor among the $2^{\delta_i}-1$.
Evidently, $C(n)$ is dominated by the cost of computing the discrete logarithm in the field $\mathbb{F}_{2^{\Delta}}$.
We remark
that the cost of computing discrete logarithms in fields of small characteristics
has recently dropped down to quasi-polynomial complexity (see \cite{joux2013new}).
We are left to estimate the number of cycles.

To estimate the number of cycles which Algorithm 1 needs to go through before stopping,
we make the following
\begin{assumption}
All chains of Zech's Logarithm computation $Z_\alpha$ with random input produce an output which is
uniformly random distributed over $\{1,2,\ldots,M-1\}$, where $M$ is the order of $\alpha$.
We assume also that for $\alpha \neq \beta$ the outputs of $Z_\alpha$ and $Z_\beta$ are independent random variables.
\end{assumption}
Although the Statistical Assumption 1 is reasonable (and experimentally verifiable), there is an unfortunate case
when
there is a pair $(i,j)$ where $\GCD(N_i,N_j) \neq 1$ and there exist no multiple of $p_i \cdot p_j$ of weight $w$
(see discussion in Subsection \ref{sss:noncoprimes}). 

Under this assumption, we can compute the $e_{i,w}$'s once every $m_i$ cycles (i.e. always if $m_i = 1$).
Furthermore, the $e_{i,w}$'s are uniformly distributed over $\{1,2,\ldots,N_i-1\}$
and so, by the Chinese Remainder Theorem, $e_w$ can be computed
with a probability $P$ where
\[
P \leq \Pi_{(i,j),i \neq j} \frac{1}{\GCD(N_i,N_j)},
\]
and it will be uniformly distibuted over $\{1,2,\ldots,N-1\}$, where $N := \LCM(\{N_i\}_{i=1,\ldots, r})$.

The condition that $2^t e_w \leq m  M$ is thus verified 
after a number of cycles estimated as
\[
O\left(\frac{1}{P} \Pi_i m_i \cdot 2^t N/(m \cdot D)\right)
\]
and
the full time complexity of Algorithm 1 is hence
\begin{equation} \label{eq:complexity}
O\left(r(w-2)(1+C(n)) \frac{1}{P} \Pi_i m_i \cdot \frac{2^t N}{m \cdot D}\right) \, .
\end{equation}

We note that can consider $P$ as constant with respect to $n$, while $N$ can be approximated by $2^n$. 
In the next subsubsection we will express $D$ as a function of $n$ and compare estimate \eqref{eq:complexity}
with the complexity of birthday-based approaches.

\subsubsection{Comparison with other approaches} \label{sss:complex2}

Birthday-based approaches compute their complexities under a different statistical assumption:
\begin{assumption}
The multiples of degree at most $D$ of a polynomial of degree $n$ has weight $w$ with probability
circa $\binom{D}{w-1} 2^{-D}$.
\end{assumption}
Using this assumption, one can prove (\cite{golic1996computation}) that
the expected number of the polynomials multiples of weight $w$ for large $D$ is:
\begin{equation} \label{eq:expNnw}
N_{n,w} \simeq \frac{\binom{D}{w-1}}{2^n} \simeq \frac{D^{w-1}}{(w-1)!2^n}
\end{equation}
Using \eqref{eq:expNnw} one computes the (expected)
critical degree where polynomials of degree $w$ will start to appear as
\begin{equation} \label{eq:D0}
D_0 \simeq (w-1)!^\frac{1}{w-1} \cdot 2^\frac{n}{w-1}
\end{equation}

A method based on the birthday paradox (\cite{meier1989fast}) finds
a multiple of weight $w$ and minimal degree $D_0$ with time complexity
\begin{equation}
O(D_0^{\lceil\frac{w-1}{2}\rceil}),
\end{equation}
and memory complexity
\begin{equation}
O(D_0^{\lfloor\frac{w-1}{2}\rfloor}).
\end{equation}

Using generalized birthday problem (\cite{wagner2002generalized}) one aims at finding
a multiple of same weight but higher degree ($D_1 \simeq 2^\frac{n}{1 + \lfloor \log (w-1) \rfloor}$) in less time.
The time complexity for this approach is
\begin{equation}
O((w-1) \cdot D_1)
\end{equation}
while memory complexity is $O(D_1)$.

In Table \ref{tab:complexity} we summarize the comparison of our discrete log approach with birthday based approaches.
Birthday based approaches do not distinguish the polynomial $p$ in terms of its factorization while Algorithm 1 does,
so for our comparison we assumed to be in the situation of $p$ a multiple of primitive polynomials whose degrees are coprime.
To simplify we also kept out the constant terms and we compare only the part of complexity which depends strictly on $n$.
We set as $D$ for comparison the $D_0$ and $D_1$ which are defined --- respectively --- by birthday and generalized birthday approaches.
We can see that Algorithm 1 performs better (same time complexity but much better memory complexity)
than both birthday based approaches in the case of weights $w=3,4$ and it performs worse than both of them
$w=5$ (higher time complexity).
We remark that this comparison does not take into account the computation constant term given by computation of discrete logarithms.
When it comes to practical implementation of the method this constant term might slow down too much Algorithm 1.
A possible solution could be to precompute discrete logarithms. This would imply an increased memory complexity of the order
of $2^q-1$ where $q$ is the biggest prime factor in the degrees of the polynomials $p_1,\ldots,p_r$.

\begin{table}

 \begin{tabular}{c|c|c|c|c|c|c|c|c|c|c|c|}
  \cline{2-12}
   & \multicolumn{3}{|c|}{Birthday} & \multicolumn{3}{|c|}{Generalized Birthday} &\multicolumn{5}{|c|}{Algorithm 1} \\ \hline
   \multicolumn{1}{|c|}{$w$} & 3 & 4 & 5 & 3 & 4 & 5 & \textbf{3} & \textbf{4} & \textbf{4} & 5 & 5 \\ \hline
   \multicolumn{1}{|c|}{$D$} & $2^{n/2}$ & $2^{n/3}$ & $2^{n/4}$ & $2^{n/2}$ & $2^{n/2}$ & $2^{n/3}$ & $2^{n/2}$ & $2^{n/2}$ & $2^{n/3}$& $2^{n/3}$ & $2^{n/4}$ \\
   \hline
   \multicolumn{1}{|c|}{time} & $2^{n/2}$ & $2^{2n/3}$ & $2^{n/2}$ & $2^{n/2}$ & $2^{n/2}$ & $2^{n/3}$ & $2^{n/2}$ & $2^{n/2}$ & $2^{2n/3}$& $2^{2n/3}$ & $2^{3n/4}$ \\
   \hline
   \multicolumn{1}{|c|}{memory} & $2^{n/2}$ & $2^{2n/3}$ & $2^{n/2}$ & $2^{n/2}$ & $2^{n/2}$ & $2^{n/3}$ & $1$ & $1$ & $1$& $1$ & $1$ \\
   \hline
 \end{tabular}


 \caption{Comparing complexity of Algorithms. We compare only the part of complexity which is exponential in the degree of the polynomial of which we want to find the multiple.
 Algorithm 1 performs better for weigth $3$ and for weight $4$.}
 \label{tab:complexity}
\end{table}

Previous discrete log based approaches were limited to $p$ being a single primitive polynomial and cannot be extended
straightforwardly to the case of general $p$.

In \cite{kuhn1994using}, Penzhorn and Kuhn used a different statistical assumption on the output of Zech's Logarithms.
They assumed that the difference of two Zech's Logarithms over $\{1,2,\ldots,N-1\}$ has exponential distribution of parameter $D/N$ when
both inputs are randomly distributed over $\{1,2,\ldots,D-1\}$. This is applied only for the case $w=4$ and gives an algorithm which is able
to compute a multiple of degree $2^{n/3}$ with time complexity $2^{n/3}$ outperforming all approaches mentioned above.

In \cite{didier2007finding}, Didier and Laigle-Chapuy used a discrete logarithm approach with time memory trade off to compute
a polynomial multiple of degree $D$ in time complexity $O(D^{\lceil \frac{w-2}{2} \rceil})$ and memory complexity $O(D^{\lfloor \frac{w-2}{2} \rfloor})$.

\subsection{Example case: E0} \label{sss:e0}

In \cite{lu2004faster}, Lu and Vaudenay describe a correlation attack to E0, the stream cipher of the Bluetooth protocol (\cite{gehrmann2004bluetooth}). 
E0 is a nonlinear combination generator with memory and it is based on four LFSRs of degrees $25,31,33,39$ with the following
feedback polynomials:
\begin{align}
 p_1 & =  x^{25} + x^{20} + x^{12} + x^8 + 1 \\
 p_2 & =  x^{31} + x^{24} + x^{16} + x^{12} + 1 \\
 p_3 & =  x^{33} + x^{28} + x^{24} + x^4 + 1 \\
 p_4 & =  x^{39} + x^{36} + x^{28} + x^4 + 1 
\end{align}
In the key-recovery attack of \cite{lu2004faster}, two precomputation steps require to compute a weight-5 multiple of
$p_2 \cdot p_3 \cdot p_4$ of degree at most $2^{34.3}$ using generalized birthday approach,
and a weight-3 multiple of $p_3 \cdot p_4$ of degree at most $2^{36}$ using a standard birthday approach.
To find the first multiple precomputation, Algorithm 1 would be inconvenient since it would require a higher time complexity.
For the second multiple precomputation, Algorithm 1 could be useful since it requires approximately the same time complexity while
having a much lower memory complexity.
We remark that in \cite{lu2004faster} the memory complexity for the precomputation part is not explicitly
stated. 

Note also that all $p_i$'s are primitive and that $\GCD(N_3,N_4) = 2^3 - 1 = 7 \neq 1$.
We are in the case where some of the $N_i$'s are not coprime
and for this specific polynomials, it is experimentally verified that the $\CRT$ can be computed roughly once every seven times.

\section{Conclusions} \label{sec:conclusions}

The algorithm presented in this paper is the first able to solve
Problem 1 in a general situation, taking
advantage of the polynomial factorization.
Its complexity should be compared to that of other general purpose
algorithms, in particolar with those
based on the (general) birthday approach. We have been able to show
that in some interesting cases
our algorithm has a time complexity comparable to the generalized
birthday approach,
while having a much lower memory complexity (i.e. $O(1)$).
These cases are relevant to the so-called faster correlation attacks
to a class of stream ciphers, which including the Bluetooth cipher E0.

\bibliographystyle{ieeetr}

\bibliography{pts-refs}

\end{document}